\documentclass[journal]{IEEEtran}
\usepackage{amssymb,stmaryrd,amsmath,amsfonts,rotating,mathrsfs}
\usepackage[noadjust]{cite}
\usepackage{color}
\usepackage{epic}
\RequirePackage{bbm}
\definecolor{TODO}{rgb}{0.6,0.6,0.6} 

\definecolor{TOCHECK}{rgb}{0.8,0.8,0.8} 

\newtheorem{theorem}{Theorem}

\newcommand{\btheo}{\begin{theorem}}
\newcommand{\etheo}{\end{theorem}}
\newcommand{\bproof}{\begin{proof}}
\newcommand{\eproof}{\end{proof}}
\newtheorem{definition}[theorem]{Definition}
\newcommand{\bdefi}{\begin{definition}}
\newcommand{\edefi}{\end{definition}}
\newtheorem{fact}[theorem]{Fact}
\newcommand{\bprop}{\begin{fact}}
\newcommand{\eprop}{\end{fact}}
\newtheorem{corollary}[theorem]{Corollary}
\newcommand{\bcor}{\begin{corollary}}
\newcommand{\ecor}{\end{corollary}}
\newtheorem{example}[theorem]{Example}
\newcommand{\bex}{\begin{example}}
\newcommand{\eex}{\end{example}}
\newtheorem{lemma}[theorem]{Lemma}
\newcommand{\blemma}{\begin{lemma}}
\newcommand{\elemma}{\end{lemma}}
\newtheorem{remark}[theorem]{Remark}
\newcommand{\bremark}{\begin{remark}}
\newcommand{\eremark}{\end{remark}}
\newtheorem{conj}[theorem]{Conjecture}
\newcommand{\bconj}{\begin{conj}}
\newcommand{\econj}{\end{conj}}

\newcommand{\reals}{\ensuremath{\mathbb{R}}}
\newcommand{\naturals}{\ensuremath{\mathbb{N}}}

\def\0{{\tt 0}} 
\def\1{{\tt 1}} 
\def\?{{\tt *}} 
\renewcommand{\mid}{\,|\,}

\newcommand{\BSC}{\ensuremath{\text{BSC}}}

\newcommand{\EEx}{\hfill $\Diamond$}

\newcommand {\p}{\tt P}
\newcommand {\scd}{\tt SC}

\newcommand {\calc}{\mathcal W}

\allowdisplaybreaks
\begin{document}
\title{The Compound Capacity of Polar Codes}
\author{S. Hamed Hassani, Satish Babu Korada and R\"udiger Urbanke \thanks{EPFL, School of Computer
 \& Communication Sciences, Lausanne, CH-1015, Switzerland,
\{seyedhamed.hassani, satish.korada, rudiger.urbanke\}@epfl.ch.
}
}

\maketitle
\begin{abstract}
We consider the compound capacity of polar codes {\em under successive
cancellation decoding} for a collection of binary-input memoryless output-symmetric
channels. By deriving a sequence of upper and lower bounds, we show that
in general the compound capacity under successive decoding is strictly
smaller than the unrestricted compound capacity.  \end{abstract}

\section{History and Motivation}
Polar codes, recently introduced by Ar\i kan \cite{Ari08}, are a family
of codes that achieve the capacity of a large class of channels using
low-complexity encoding and decoding algorithms. The complexity of
these algorithms scales as $O(N\log N)$, where $N$ is the blocklength
of the code. Recently, it has been shown that, in addition to being
capacity-achieving for channel coding,  polar-like codes are also optimal
for lossy source coding as well as multi-terminal problems like the
Wyner-Ziv and the Gelfand-Pinsker problem \cite{KoU09}.

Polar codes are closely related to Reed-Muller (RM) codes. The rows of the generator
matrix of a polar code of length $N=2^n$ are chosen from the rows
of the matrix  $G^{\otimes n}=\bigl[ \begin{smallmatrix} 1 &0 \\ 1&
1 \end{smallmatrix} \bigr]^{\otimes n}$, where $\otimes$ denotes the
Kronecker product. The crucial difference of polar codes to RM codes
is in the choice of the rows. For RM codes the rows of largest weight
are chosen, whereas for polar codes the choice is dependent on the
channel. We refer the reader to \cite{Ari08} for a detailed discussion on
the construction of polar codes. The decoding is done using a successive
cancellation (SC) decoder. This algorithm decodes the bits one-by-one in a pre-chosen
order.

Consider a communication scenario where the transmitter and the receiver
do not know the channel. The only knowledge they have is the set
of channels to which the channel belongs. This is known as the {\em
compound channel} scenario. Let $\calc$ denote the set of channels. 
The compound capacity of $\calc$ is defined as the rate at which we
can reliably transmit irrespective of the particular channel (out of
$\calc$) that is chosen. The compound capacity is given by \cite{BBT59}
\begin{align*}
C(\calc) = \max_{P}\inf_{W\in\calc}I_P(W),
\end{align*}
where $I_P(W)$ denotes the mutual information between the input and
the output of $W$, with the input distribution being $P$. Note that the
compound capacity of $\calc$ can be strictly smaller than the infimum
of the individual capacities. This happens if the capacity-achieving
input distribution for the individual channels are different. On the other
hand, if the capacity-achieving input distribution is the same for all
channels in $\calc$, then the compound capacity is equal to the infimum
of the individual capacities. This is indeed the  case
since we restrict our attention to the class of binary-input memoryless
output-symmetric (BMS) channels.

We are interested in the maximum achievable rate using polar codes and SC
decoding. We refer to this as the compound capacity using polar codes and
denote it as $C_{\p, \scd}(\calc)$.  More precisely, given a collection
$\calc$ of BMS channels we are interested in constructing a polar code of
rate $R$ which works well (under SC decoding) for every channel in this
collection. This means, given a target block error probability, call it
$P_B$, we ask whether there exists a polar code of rate $R$ such that its
block error probability is at most $P_B$ for any channel in $\calc$. In
particular, how large can we make $R$ so that a construction exists for
any $P_B > 0$?

We consider the compound capacity with respect to ignorance at the
transmitter but we allow the decoder to have knowledge of the actual
channel.

\section{Basic Polar Code Constructions}
Rather than describing the standard construction of polar codes, let us
give here an alternative but entirely equivalent formulation. For the
standard view we refer the reader to \cite{Ari08}.

Binary polar codes have length $N=2^n$, where $n$ is an integer.
Under successive decoding, there is a BMS channel associated to each bit
$U_i$ given the observation vector $Y_0^{N-1}$ as well as the values
of the previous bits $U_0^{i-1}$.  This channel has a fairly simple
description in terms of the underlying BMS channel $W$.\footnote {We
note that in order to arrive at this description we crucially use the
fact that $W$ is symmetric. This allows us to assume that $U_0^{i-1}$
is the all-zero vector.}

\begin{definition}[Tree Channels of Height $n$]
Consider the following $N=2^n$ tree channels of height $n$. Let
$\sigma_{1}\dots \sigma_n$ be the $n$-bit binary expansion of
$i$. E.g., we have for $n=3$, $0=000$, $1=001$, \dots, $7=111$. Let 
$\sigma = \sigma_1\sigma_2\dots\sigma_{n}$. Note
that for our purpose it is slightly more convenient to denote the least
(most) significant bit as $\sigma_n$ ($\sigma_1$).  Each tree
channel consists of $n+1$ levels, namely $0,\dots,n$. It is a complete
binary tree. The root is at level $n$. At level $j$ we have $2^{n-j}$
nodes. For $1 \leq j \leq n$, if $\sigma_{j} = 0$ then all nodes on level
$j$ are check nodes; if $\sigma_{j} = 1$ then all nodes on level $j$
are variable nodes. All nodes at level $0$ correspond to independent
observations of the output of the channel $W$, assuming that the input
is $0$.  
\end{definition}

An example for $W^{011}$ (that is $n=3$ and $\sigma=011$) is shown in
Figure~\ref{fig:tree}.
\begin{figure}[!h] \begin{center} \input{treenew1.tex} \end{center}
\caption{ Tree representation of the channel $W^{011}$. The $3$-bit binary
expansion of $3$ is $\sigma_1\sigma_2\sigma_3 = 011$.}\label{fig:tree}
\end{figure}

Let us call $\sigma=\sigma_{1}\dots\sigma_n$ the {\em type} of the
tree. We have $\sigma \in \{0, 1\}^n$.  Let $W^{\sigma}$ be the
channel associated to the tree of type  $\sigma$.  Then $I(W^{\sigma})$
denotes the corresponding capacity. Further, by $Z(W^{\sigma})$
we mean the corresponding Bhattacharyya functional (see  \cite[Chapter 4] {RiU08}).

Consider the channels $W_N^{(i)}$ introduced by Ar\i kan in \cite{Ari08}.
The channel $W_N^{(i)}$ has input $U_i$ and output $(Y_0^{N-1}, U_0^{i-1})$.
Without proof we note that $W_N^{(i)}$ is equivalent to the
channel $W^{\sigma}$ introduced above if we let $\sigma$ be the
$n$-bit binary expansion of $i$.

Given the description of $W^{\sigma}$ in terms of a tree channel,
it is clear that we can use density evolution \cite{RiU08} to compute the
channel law of $W^{\sigma}$. Indeed, assuming that infinite-precision
density evolution has unit cost, it was shown in \cite{MoT09} that the
total cost of computing all channel laws is linear in $N$. 

When using density evolution it is convenient to represent the channel
in the log-likelihood domain. We refer the reader to \cite{RiU08} for a
detailed description of density evolution. The BMS $W$ is represented
as a probability distribution over $\mathbb{R}\cup\{\pm\infty\}$. The
probability distribution is the distribution of the variable
$\log(\frac{W(Y\mid 0)}{W(Y\mid 1)})$, where $Y\sim W(y\mid 0)$.

Density evolution starts at the leaf nodes which are
the channel observations and proceeds up the tree. We have two
types of convolutions, namely the variable convolution (denoted by
$\circledast$) and the check convolution (denoted by $\boxast$). All
the densities corresponding to nodes which are at the same level are
identical. Each node in the $j$-th level is connected to two nodes in
the $(j-1)$-th level. Hence the convolution (depending on $\sigma_j$)
of two identical densities in the $(j-1)$-th level yields the density in
the $j$-th level. If $\sigma_j=0$, then we use a check convolution ($\boxast$), and if
$\sigma_j=1$, then we use a variable convolution ($\circledast$).  

\begin{example}[Density Evolution] Consider the channel shown in Figure~\ref{fig:tree}. By
some abuse of notation, let $W$ also denote the initial density
corresponding to the channel $W$. Recall that $\sigma=011$.  Then the
density corresponding to $W^{011}$ (the root node) is given by
\begin{align*}
\Bigl((W^{\boxast 2})^{\circledast 2}\Bigr)^{\circledast 2} = (W^{\boxast 2})^{\circledast 4}.
\end{align*}
\EEx
\end{example}

\section{Main Results}
Consider two BMS channels $P$ and $Q$. We are interested in constructing
a common polar code of rate $R$ (of arbitrarily large block length)
which allows reliable transmission over both channels.

Trivially,
\begin{align} \label{equ:Ibound}
C_{\text{P, SC}}(P, Q) & \leq \min\{I(P), I(Q)\}.
\end{align}
We will see shortly that, properly applied, this simple fact
can be used to give tight bounds.

For the lower bound we claim that
\begin{align} 
C_{\text{P, SC}}(P, Q) & \geq C_{\text{P, SC}}(\text{BEC}(Z(P)), \text{BEC}(Z(Q))) \nonumber \\
                       & = 1-\max\{Z(P), Z(Q)\}.   \label{equ:Zbound}
\end{align}
To see this claim, we proceed as follows. Consider a particular
computation tree of height $n$ with observations at its leaf nodes
from a BMS channel with Battacharyya constant $Z$. What is the largest
value that the Bhattacharyya constant of the root node can take on?
From the extremes of information combining framework (\cite[Chapter 4]
{RiU08}) we can deduce that we get the largest value if we take the
BMS channel to be the BEC$(Z)$. This is true, since at variable nodes
the Bhattacharyya constant acts multiplicatively for any channel, and
at check nodes the worst input distribution is known to be the one from
the family of BEC channels. Further, BEC densities stay preserved within
the computation graph.

The above considerations give rise to the following transmission scheme.
We signal on those channels $W^{\sigma}$ which are reliable for the
BEC$(\max\{Z(P), Z(Q)\})$. A fortiori these channels are also reliable
for the actual input distribution. In this way we can achieve a reliable
transmission at rate $1-\max\{Z(P), Z(Q)\}$.

\begin{example}[BSC and BEC]
Let us apply the above mentioned bounds to $C_{\text{P, SC}}(P,
Q)$, where $P=\text{BEC}(0.5)$ and $Q=\text{BSC}(0.11002)$. We
\begin{align*}
I(P)=I(Q)& =0.5, \\
Z(\text{BEC}(0.5)) & =0.5, \\
Z(\text{BSC}(0.11002))& = 2 \sqrt{0.1102 (1-0.11002)} \approx 0.6258.
\end{align*}
The upper bound (\ref{equ:Ibound}) and the lower bound (\ref{equ:Zbound}) 
then translate to
\begin{gather*}
C_{\text{P, SC}}(P, Q))  \leq \min\{0.5,0.5\} = 0.5, \\
C_{\text{P, SC}}(P, Q))  \geq 1 -\max\{0.6258, 0.5\} = 0.3742.
\end{gather*}
Note that the upper bound is trivial, but the lower bound is not.
\EEx
\end{example}

In some special cases the best achievable rate is easy to determine.
This happens in particular if the two channels are ordered by
degradation.  
\begin{example}[BSC and BEC Ordered by Degradation]
Let $P=\text{BEC}(0.22004)$ and $Q=\text{BSC}(0.11002)$.  We have
$I(P)=0.770098$ and $I(Q)=0.5$.  Further, 
one can check that the $\text{BSC}(0.11002)$ is degraded
with respect to the $\text{BEC}(0.22004)$. This implies
that any sub-channel of type $\sigma$ which is good for the
$\text{BSC}(0.11002)$, is also good for the $\text{BEC}(0.22004)$.  Hence,
$C_{\text{P,SC}}(\text{BEC}(0.22004),\text{BSC}(0.11002))=I(Q)=0.5$. 
\EEx
\end{example}

More generally, if the channels $\mathcal W$ are such that there is a channel
$W\in\mathcal W$ which is degraded with respect to every channel in $\mathcal
W$, then $C_{\text{P,SC}}(\mathcal W) = C (\mathcal W) = I(W)$. Moreover, the sub-channels $\sigma$
that are good for $W$ are good also for all channels in $\mathcal W$.

So far we have looked at seemingly trivial upper and lower bounds on the compound
capacity of two channels.  As we will see now, it is quite simple to considerably tighten
the result by considering individual branches of the computation tree separately.
\begin{theorem}[Bounds on Pairwise Compound Rate]
\label{the:compoundbound}
Let $P$ and $Q$ be two BMS channels. Then for any $ n \in \naturals$

\begin{align*}
C_{\text{P, SC}}(P, Q)  \leq & 
          \frac{1}{2^n} \sum_{\sigma \in \{0, 1\}^n}  \min\{ I(P^{\sigma}), I(Q^{\sigma})   \},  \\
C_{\text{P, SC}}(P, Q)  \geq  
&         1- \frac{1}{2^n} \sum_{\sigma \in \{0, 1\}^n}  \max\{ Z(P^{\sigma}), 
Z(Q^{\sigma})\}.
\end{align*}
Further, the upper as well as the lower bounds converge to the compound capacity as
$n$ tends to infinity and the bounds are monotone with respect to $n$. 
\end{theorem}
\begin{proof}
Consider all $N=2^n$ tree channels. Note that there are $2^{n-1}$ such
channels that have $\sigma_{1}=0$  and $2^{n-1}$ such channels that have
a $\sigma_{1}=1$. Recall that $\sigma_{1}$ corresponds to the type of
node at level $n$.

This level transforms the original channel $P$ into $P^{0}$ and
$P^{1}$, respectively. Consider first the $2^{n-1}$ tree channels
that correspond to $\sigma_{1}=1$. Instead of thinking of each tree as
a tree of height $n$ with observations from the channel $P$, think of
each of them as a tree of height $n-1$ with observations coming from
the channel $P^{1}$.  By applying our previous argument, we see
that if we let $n$ tend to infinity then the common capacity for this
half of channels is at most $0.5 \min\{I(P^{1}), I(Q^{1})\}$.
Clearly the same argument can be made for the second half of channels.
This improves the trivial upper bound (\ref{equ:Ibound}) to
\begin{align*}
C_{\text{P, SC}}(P, Q)  \leq & 0.5 \min\{I(P^{1}), I(Q^{1})\} + \\
& 0.5 \min\{I(P^{0}), I(Q^{0})\}.
\end{align*}
Clearly the same argument can be applied to trees of any height $n$.
This explains the upper bound on the compound capacity
of the form $\min\{I(P^{\sigma}), I(Q^{\sigma})\}$.

In the same way we can apply this argument to the lower bound (\ref{equ:Zbound}).

From the basic polarization phenomenon we know that for every $\delta>0$
there exists an $n \in \naturals$ so that
\begin{align*}
\frac{1}{2^n} |\{\sigma \in \{0, 1\}^n: I(P^{\sigma}) \in [\delta, 1-\delta] \}| \leq \delta/4.
\end{align*} 
Equivalent statements hold for $I(Q^{\sigma})$, $Z(P^{\sigma})$, and $Z(Q^{\sigma})$.

In words, except for at most a fraction $\delta$, all channel pairs $(P^{\sigma}, Q^{\sigma})$
have ``polarized." For each polarized pair both the upper as well as the lower bound
are loose by at most $\delta$. Therefore, the gap between the upper and lower bound
is at most $(1-\delta) 2 \delta + \delta$.

To see that the bounds are monotone consider a particular type $\sigma$ of length $n$.
Then we have
\begin{align*}
& \min \{I(P^{\sigma}), I(Q^{\sigma})\} \\
& = \min \{\frac12 (I(P^{\sigma 0})+I(P^{\sigma 1})), \frac12 (I(Q^{\sigma 0})+I(Q^{\sigma 1}))\} \\
&\geq \frac 12 \min \{I(P^{\sigma 0}), I(Q^{\sigma 0})\} + \frac12 \min \{I(P^{\sigma 1}), I(Q^{\sigma 1})\}. 
\end{align*}
A similar argument applies to the lower bound.
\end{proof}

Remark: 
In general there is no finite $n$ so that either upper or lower
bound agree exactly with the compound capacity. 
On the positive side, the lower bounds are constructive and give an actual strategy to
construct polar codes of this rate.
\begin{example}[Compound Rate of BSC$(\delta)$ and BEC$(\epsilon)$]
Let us compute upper and lower bounds on $C_{\text{P, SC}}(\text{BSC}(0.11002), \text{BEC}(0.5))$. 
Note that both the BSC$(0.11002)$ as well as the BEC$(0.5)$ have capacity one-half.
Applying the bounds of Theorem~\ref{the:compoundbound} we get:
\begin{center}
\begin{tabular}{ccccccc} 
n=0     & 1     & 2     & 3     & 4     & 5     & 6 \\ \hline
0.500 & 0.482 & 0.482 & 0.482 & 0.482 & 0.482 & 0.482 \\
0.374 & 0.407 & 0.427 & 0.440 & 0.449 & 0.456 & 0.461
\end{tabular}
\end{center}
These results suggest that the numerical value of 
$C_{\text{P, SC}}(\text{BSC}(0.11002), \text{BEC}(0.5))$ is close to
$0.482$.
\EEx
\end{example}

\begin{example}[Bounds on Compound Rate of BMS Channels]
In the previous example we considered the compound capacity of two BMS
channels. How does the result change if we consider a whole family of
BMS channels. E.g., what is $C_{\text{P, SC}}(\{\text{BMS}(I=0.5)\})$?

We currently do not know of a procedure (even numerical) to compute this
rate. But it is easy to give some upper and lower bounds.

In particular we have
\begin{align}
C_{\text{P, SC}}(\{\text{BMS}(I=0.5)\}) & \leq C(\text{BSC}(0.11002), \text{BEC}(0.5)) \nonumber \\
                                           & \leq 0.4817, \nonumber \\
C_{\text{P, SC}}(\{\text{BMS}(I=0.5)\}) & \geq 1-Z(\text{BSC}(I=0.5)) \approx 0.374. 
\label{equ:universallowerbound} 
\end{align}
The upper bound is trivial. The compound rate of a whole class cannot be
larger than the compound rate of two of its members. For the lower bound note
that from Theorem~\ref{the:compoundbound} we know that the achievable
rate is at least as large as $1-\max\{Z\}$, where the maximum is over
all channels in the class. Since the BSC has the largest Bhattacharyya
parameter of all channels in the class of channels with a fixed capacity,
the result follows.

\EEx
\end{example}

\section{A Better Universal Lower Bound}\label{sec:universal}
The universal lower bound expressed in (\ref{equ:universallowerbound})
is rather weak.  Let us therefore show how to strengthen it.

Let $\mathcal W$ denote a class of BMS channels.
From Theorem~\ref{the:compoundbound} we know that
in order to evaluate the lower bound we have to optimize the terms
$Z(P^{\sigma})$ over the class $\mathcal W$.

To be specific, let $\mathcal W$ be $\text{BMS}(I)$, i.e., the space of
BMS channels that have capacity $I$. Expressed in an alternative way,
this is the space of distributions that have entropy equal to $1-I$. 

The above optimization is in general a difficult problem.
The first difficulty is that the space $\{\text{BMS}(I)\}$ is infinite
dimensional. Thus, in order to use numerical procedures we have to 
approximate this space by a finite dimensional space. Fortunately, as
the space is compact, this task can be accomplished. E.g.,  
look at the densities corresponding to the class $\{\text{BMS}(I)\}$ in the $\vert D \vert
$-domain. In this domain, each BMS channel $W$ is represented by the
density corresponding to the probability distribution of $\vert W(Y\mid
0) - W(Y\mid 1)\vert$, where $Y \sim W(y\mid 0)$. For example, the
$|D|$-density corresponding to $\BSC(\epsilon)$ is $\Delta_{1-2\epsilon}$.

We quantize the interval $[0, 1]$ using real values $0= p_1 < p_2 <
\cdots < p_m =1$, $m \in \naturals$. The $m$-dimensional polytope
approximation of $\{\text{BMS}(I)\}$, denoted by ${\mathcal W}_m$, is the space of
all the densities which are of the form $\sum_{i=1} ^{m} \alpha_i \Delta
_{p_i}$. Let $\alpha = [\alpha_1,\cdots,\alpha_m]^\top$. Then $\alpha$
must satisfy the following linear constraints:
\begin{equation}
\begin{aligned}
\alpha^\top 1_{m\times 1}=1, \; \alpha^\top H_{m\times 1} = 1-I,\; & \alpha_i \geq 0,
\end{aligned}
\label{quantization}
\end{equation}
where $H_{m\times1}=[h_2 (\frac{1-p_i}{2})]_{m\times1}$ and $1_{m\times1}$ is the
all-one vector. 

Due to quantization, there is in general an approximation error. 

\begin{lemma}[$m$ versus $\delta$] \label{lem:approximation}
Let $a \in \text{BMS}(I)$. Assume a uniform quantization of the interval $[0, 1]$
with $m$ points $0 = p_1 < p_2 < \dots < p_m=1$. If $m \geq 1+  \frac {1}
{1-\sqrt[4]{1-\delta^2}}$, then there exists a density $b \in {\mathcal W}_m$
such that $\vert Z(a\boxast a)-Z(b\boxast b) \vert \leq \delta$.
\end{lemma}
\begin{proof}
For a given density $a$, let $Q_u(a) (Q_d(a))$ denote the
quantized density obtained by mapping the mass in the interval
$(p_i,p_{i+1}]$($[p_i,p_{i+1})$) to $p_{i+1}$ ($p_i$).  Note that $Q_u(a)$
($Q_d(a)$) is {\em u}pgraded ({\em d}egraded) with respect to $a$.  Thus, $H(Q_u(a))
\leq H(a) \leq H(Q_d(a))$. The Bhattacharyya parameter $Z(a\boxast a)$is given by
$$Z(a\boxast a)=\int_0 ^{1} \int_0^1\sqrt {1-x_1^2x_2^2} a(x_1) dx_1 a(x_2) dx_2.$$
Since $\sqrt {1-x^2}$ is
decreasing on $[0,1]$, we have
\begin{align*}
& Z(Q_d(a)\boxast Q_d(a))-Z(a\boxast a) \\
&\leq \sum_{i,j=1}^{m-1} \int_{p_i} ^{p_{i+1}}  \int_{p_j} ^{p_{j+1}}\Bigl(\sqrt
{1-p_i^2p_j^2}-\sqrt {1-x^2y^2} \Bigr)\\
&\phantom{=================} a(x) dx a(y) dy,\\
& Z(a\boxast a ) - Z(Q_u(a)\boxast Q_u(a)) \\
&\leq \sum_{i,j=1}^{m-1} \int_{p_i} ^{p_{i+1}} \int_{p_j} ^{p_{j+1}} \Bigl(\sqrt
{1-x^2y^2}-\sqrt {1-p_{i+1}^2p_{j+1}^2} \Bigr)\\
&\phantom{=================}a(x) dx a(y) dy.
\end{align*}
Now note that the maximum approximation error, call it $\delta$, happens 
when $xy$ is close to $1$. This maximum error is equal to 
\begin{align*}
& \sqrt {1-\Big(1-\Big(\frac1{m-1}\Big)\Big)^4}-\sqrt {1-1^4}.
\end{align*}
Solving for $m$ we see that the quantization error can be made smaller 
than $\delta$ by choosing $m$ such that 
\begin{align} 
 m \geq 1+  \frac {1} {1-\sqrt[4]{1-\delta^2}}.
\label {m_delta}
\end{align}
Note that if $a \in {\mathcal W}$ then in general neither $Q_d(a)$
nor $Q_d(a)$ are elements of ${\mathcal W}_m$, since their entropies
do not match.  In fact, as discussed above, the entropy of $Q_d(a)$
is too high, and the entropy of $Q_u(a)$ is too low. But by taking a
suitable convex combination we can find an element $b \in {\mathcal W}_m$
for which $Z(b^{\boxast 2})$ differs from $Z(a ^{\boxast 2})$ by at most $\delta$.

In more detail, consider the function $f(t)=H(tQ_u(a)+(1-t)Q_d(a))$,
$0 \leq t \leq 1$. Clearly, $f$ is a continuous function on its
domain. Since every density of the form of $tQ_u(a)+(1-t)Q_d(a)$ is
upgraded with respect to $Q_d(a)$ and degraded with respect to $Q_u(a)$,
we have $Z((Q_u(a))^{\boxast 2})\leq Z((tQ_u(a)+(1-t)Q_d(a))^{\boxast
2}) \leq Z((Q_d(a))^{\boxast 2})$. As a result: $\vert
Z((tQ_u(a)+(1-t)Q_d(a))^{\boxast 2})-Z(a^{\boxast 2}) \vert \leq \delta $.
We further have $f(0)=H(Q_u(a))\leq H(a) \leq H(Q_d(a)) = f(1)$. Thus
there exists a $0 \leq t_0 \leq 1$ such that $f(t_0)=H(a)=I$. Hence, $t_0
Q_u(a)+(1-t_0)Q_d(a) \in \text{BMS}(I)$ and $t_0 Q_u(a)+(1-t_0)Q_d(a)
\in {\mathcal W}_m$. Therefore $t_0 Q_u(a)+(1-t_0)Q_d(a)$ is the desired
density.  \end{proof}

\begin{example}[Improved Bound for BMS$(I=\frac12)$]
Let us derive an improved bound for the class ${\mathcal W}=\text{BMS}(I=\frac12)$.
We pick $n=1$, i.e., we consider tree channels of 
height $1$ in Theorem~\ref{the:compoundbound}.

For $\sigma=0$ the implied operation is $\circledast$.  It is well
known that in this case the maximum of $Z(a \circledast a)$ over
all $a \in {\mathcal W}$ is achieved for $a = \text{BSC}(0.11002)$.
The corresponding maximum $Z$ value is $0.3916$.  

Next consider $\sigma=1$. This corresponds to the convolution $\boxast$. 
Motivated by Lemma~\ref{lem:approximation} consider at first the 
maximization of $Z$ within the class ${\mathcal W}_m$:
\begin{equation}
\begin{aligned}
& \text{maximize}: 
\sum_{i,j} \alpha_{i} \alpha_{j} Z(\Delta_{p_{i}} \boxast \Delta_{p_{j}}) =
\sum_{i,j} \alpha_{i} \alpha_{j} \sqrt{1-(p_{i} p_{j})^2} \\
& \text{subject to : } \alpha^\top 1_{m\times 1}=1, \; \alpha^\top H_{m\times 1} = \frac 12,\; 
\alpha_i \geq 0.
\end{aligned}
\end{equation}
In the above, since the $p_i$s are fixed, the terms
$\sqrt{1-(p_{i} p_{j})^2}$ are also fixed. The
task is to optimize the quadratic form
$\alpha^\top P \alpha$ over
the corresponding $\alpha$ polytope, where the $m\times m$ matrix $P$ is
defined as $P_{ij}=\sqrt{1-(p_i p_j)^2}$. 
We claim that this is a convex optimization problem. 

To see this, expand $\sqrt{1-x^2}$ as a Taylor series in the form
\begin{equation}
\sqrt{1-x^2}=1-\sum_{l \geq 0} t_l x^{2l},
\end{equation}
where the $t_l\geq 0$. We further have
\begin{equation}
\alpha^{\top}P \alpha=\sum_{i,j} \alpha_i \alpha_j \sqrt{1-(p_i p_j)^2}= 
1-\sum_{l\geq0} t_l \Bigl(\sum_i \alpha_i {p_i}^{2l}\Bigr)^2.
\end{equation}
Thus, since $t_l \geq 0$ and the $p_i$s are fixed, each of the terms 
$-t_l( \sum_i \alpha_i {p_i}^{2l})^2$ in the above sum represents a concave function. As
a result the whole function is concave. 


To find a bound, let us relax the condition $0 \leq \alpha_i \leq 1$
and admit $\alpha \in \reals$. We are thus faced with solving the convex
optimization problem 
\begin{equation}
\begin{aligned}
& \text{maximize}: \alpha^{\top}P \alpha \nonumber \\
& \text{subject to : } \alpha^\top 1_{m\times 1}=1, \; \alpha^\top H_{m\times 1} = \frac 12.
\end{aligned}
\end{equation}
The Kuhn-Tucker conditions for this problem yield 
\begin{equation}
\begin{bmatrix}
2P & 1 & H \\
1^\top & 0 & 0 \\
H^\top& 0 & 0 
\end{bmatrix}
\begin{bmatrix}
 \alpha_1 \\
 \alpha_2 \\
\vdots \\
\alpha_n\\
\lambda_1\\
\lambda_2
\end{bmatrix}=
\begin{bmatrix}
0 \\
0 \\
\vdots \\
0\\
1\\
\frac 12
\end{bmatrix}.
\end{equation}
As $P$ is non-singular, the answer to the above set of linear equations is
unique.  

We can now numerically compute this upper bound and from
Lemma~\ref{lem:approximation} we have an upper bound on the estimation
error due to quantization.  We get an approximate value of $0.799$.
We conclude that 
\begin{align*} C_{\text{P, SC}}(\{\text{BMS}(I=0.5)\})
&\geq 1-\frac 12 (0.392+0.799)\\ &=0.404.  
\end{align*} 
This slightly improves on the value $0.374$ in \eqref{equ:universallowerbound}.
In principle even better bounds can be derived by considering
values of $n$ beyond $1$. But the implied optimization problems 
that need to be solved are non-trivial.
\EEx
\end{example}

\section{Conclusion and Open Problems}
We proved that the compound capacity of polar codes under SC decoding
is in general strictly less than the compound capacity itself.  It is
natural to inquire why polar codes combined with SC decoding fail to
achieve the compound capacity. Is this due to the codes themselves
or is it a result of the sub-optimality of the decoding algorithm?
We pose this as an interesting open question.

In \cite{KSU09} polar codes based on general $\ell \times \ell$ matrices
$G$ were considered. It was shown that suitably chosen such codes have
an improved error exponent. Perhaps this generalization is also useful
in order to increase the compound capacity of polar codes.

\section*{Acknowledgment}
The work presented in this paper is partially supported by the National
Competence Center in Research on Mobile Information and Communication
Systems (NCCR-MICS), a center supported by the Swiss National Science
Foundation under grant number 5005-67322 and by the Swiss National
Science Foundation under grant number 200021-121903.

\bibliographystyle{IEEEtran} 
\bibliography{lth,lthpub}
\end{document}